\documentclass[twocolumn,pra,aps,showpacs,preprintnumbers,nofootinbib, superscriptaddress, amsmath,amssymb]{revtex4}
\usepackage{amstext,mathrsfs,amsthm,bm}
\usepackage{mathtools} 

\newcommand*{\bra}[1]{\left\langle{#1}\right|} 
\newcommand*{\ket}[1]{\left| #1 \right\rangle}

\newcommand*{\commut}[1]{[ #1 ]}

\newtheorem{theorem}{Theorem}
\newtheorem{lemma}{Lemma}

\begin{document}

\author{Denys I. Bondar}
\email{dbondar@princeton.edu}
\affiliation{Department of Chemistry, Princeton University, Princeton, NJ 08544, USA} 

\author{Renan Cabrera}
\email{rcabrera@princeton.edu}
\affiliation{Department of Chemistry, Princeton University, Princeton, NJ 08544, USA} 

\author{Herschel A. Rabitz}
\email{hrabitz@princeton.edu}
\affiliation{Department of Chemistry, Princeton University, Princeton, NJ 08544, USA}

\title{Conceptual inconsistencies in finite-dimensional quantum and classical mechanics}

\begin{abstract}
Utilizing operational dynamic modeling [Phys. Rev. Lett. {\bf 109}, 190403 (2012)], we demonstrate that any finite-dimensional representation of quantum and classical dynamics violates the Ehrenfest theorems. Other peculiarities are also revealed, including the nonexistence of the free particle and ambiguity in defining potential forces. Non-Hermitian mechanics is shown to have the same problems. This work compromises a popular belief that finite-dimensional mechanics 
is a straightforward discretization of the corresponding infinite-dimensional formulation. 
\end{abstract}

\date{\today} 

	
\maketitle 

\section{Introduction} 

Schr\"{o}dinger formulated quantum mechanics in terms of differential operators acting on complex valued wave functions. Heisenberg devised the matrix representation with vectors replacing continuous functions. The latter description seems more intuitive for processes in bound systems, while the former is natural for scattering problems, where de Broglie waves represent interacting particles propagating in the continuum. As ubiquitously known, these two quantum mechanical representations are merely special realizations of the unifying Hilbert space formulation, where the observables are self-adjoint operators acting on ket-vectors that abstract the wave function's notion. Moreover, Schr\"{o}dinger and Heisenberg mechanics are equivalent as long as the underlying Hilbert space is {\it infinite-dimensional}, which is a paraphrase of the fact that all complete infinite-dimensional spaces are alike.

However graceful such a quantum mechanical formalism is, the requirement of infinite-dimensionality appears to be excessive as far as applications are concerned. For example, whenever the Schr\"{o}dinger equation is solved numerically, it is always approximated by a finite matrix equation. Hence, a self-consistent formulation of quantum mechanics in a finite-dimensional Hilbert space is an ongoing problem \cite{Weyl1950, Schwinger1960, Schwinger2003, Pearle1973, Santhanam1976, Jagannathan1981, Gudder1981, Jagannathan1982, Jagannathan1983, Gudder1985, Gudder1986, Busch1993, Galetti1999, Hakioglu2000, DelaTorre2003a, DelaTorre2003b, Vourdas2004, Gazeau2006, Revzen2008, Bang2009, Sasaki2011}, especially considering quantum information science, quantum optics, and  physics of the Planck length, where finite-dimensionality finds new horizons. 

The purpose of this work is to pinpoint fundamental inconsistencies plaguing {\it any} attempt to formulate quantum as well as classical mechanics in a finite-dimensional framework. In particular, we will demonstrate the violation of the Ehrenfest theorems, nonexistence of the free particle, and  an ambiguity in defining potential forces. Possible solutions of these inconsistencies are outlined.

\section{Mathematical Background} 

We begin by constructing the coordinate $\hat{x}$ and momentum $\hat{p}$ operators. Utilizing group-theoretic arguments, Weyl demonstrated \cite{Weyl1950} 
\begin{align}\label{PhysOrigin_CCR}
	\left| \langle x \ket{p} \right|^2 = \mbox{const} \Longleftrightarrow  \commut{ \hat{x}, \hat{p} } = i \hbar,
\end{align}
(for alternative derivations see Refs. \cite{Schwinger2003, Bondar2011a}). The left-hand side of Eq. (\ref{PhysOrigin_CCR}) is well defined  in a finite-dimensional space, whereas the right-hand side of Eq. (\ref{PhysOrigin_CCR}) is not because no bounded operators obey the canonical commutation relation \cite{Weyl1950, Reed1980}. Thus, in the literature on finite-dimensional quantum mechanics, the coordinate and momentum operators are usually defined as those whose eigenvectors obey
\begin{align}\label{Deff_PX_as_MUBS}
	\langle x_n \ket{p_k} = \exp(i2\pi n k /N) / \sqrt{N}, 
\end{align}
where $n,k = -j, -j+1, \ldots, j-1, j$, $j =2N+1$, $\hat{x} \ket{x_n} = an\ket{x_n}$, $\hat{p} \ket{p_k} = 2\pi \hbar k/(aN) \ket{p_k}$, and $N$ is the Hilbert space's dimension (see Ref. \cite{DelaTorre2003a} for a detailed description of the notation). According to Eq. (\ref{Deff_PX_as_MUBS}), the coordinate and momentum eigenvectors form mutually unbiased bases, which means that if a system is prepared in a state from one of the bases, then all outcomes of the measurement with respect to the other basis are equally probable.

The finite-dimensional Schr\"{o}dinger equation is
\begin{align}\label{Sch_Eq_General_Form}
	i\hbar \ket{d\Psi(t)/dt} = \hat{H} \ket{\Psi(t)},
\end{align}
where there are at least two independent ways of introducing $\hat{H}$. First, the Hamiltonian   
\begin{align}\label{Infinite_Dim_Form_Hamiltonian}
	\hat{H} = \hat{p}^2/(2m) + U(\hat{x})
\end{align}
can be adapted with the specification that $\hat{x}$ and $\hat{p}$ are defined by Eq. (\ref{Deff_PX_as_MUBS}). This is nearly a tacit definition of the Hamiltonian in finite-dimensional quantum mechanics. An alternative definition of $\hat{H}$,  widely employed in numerical calculations (see, e.g., Ref. \cite{Kimball1934}), is obtained once the term $\hat{p}^2/(2m)$,   proportional to the second derivative in the infinite-dimensional coordinate representation, is approximated by finite differences. These two forms of the $N$-dimensional Hamiltonians are significantly different even though they converge to the same limit as $N \to \infty$.

Which definition is more fundamental? To answer this question, we employ Operational Dynamic Modeling (ODM) \cite{Bondar2011c}, a universal and systematic framework for deriving equations of motion from the evolution of the dynamical average values. In Ref. \cite{Bondar2011c}, along with a number of other applications, we utilized this method to infer the classical Liouville and Schr\"{o}dinger equations from the Ehrenfest theorems, 
\begin{align}
	m\frac{d}{dt} \bra{\Psi(t)} \hat{x} \ket{\Psi(t)} & = \bra{\Psi(t)} \hat{p} \ket{\Psi(t)}, 	\label{EhrenfestTh1} \\
	\frac{d}{dt} \bra{\Psi(t)} \hat{p} \ket{\Psi(t)} & = \bra{\Psi(t)} -U'(\hat{x}) \ket{\Psi(t)},	\label{EhrenfestTh2}
\end{align}
by assuming that the coordinate $\hat{x}$ and momentum $\hat{p}$ commuted in the former case and obeyed the canonical commutation relation in the latter.  

\section{Revealing Inconsistencies in Finite-dimensional Quantum and Classical Mechanics}

Closely following Ref. \cite{Bondar2011c}, we apply ODM to Eqs. (\ref{EhrenfestTh1}) and (\ref{EhrenfestTh2}): Stone's theorem (see, e.g., Ref. \cite{Reed1980, Araujo2008})  
guarantees the uniqueness of the Hermitian operator $\hat{H}$ in Eq. (\ref{Sch_Eq_General_Form}). Applying the chain rule to the left-hand side of Eqs. (\ref{EhrenfestTh1}) and (\ref{EhrenfestTh2}) and then utilizing Eq. (\ref{Sch_Eq_General_Form}), one obtains
\begin{align}
	im \commut{\hat{H}, \hat{x}} &= \hbar \hat{p},  \label{First_CommutEq_for_Hamiltonian} \\
	i\commut{\hat{H}, \hat{p}} &= -\hbar U'(\hat{x}). \label{Second_CommutEq_for_Hamiltonian}
\end{align}    
We shall demonstrate below that Eqs. (\ref{First_CommutEq_for_Hamiltonian}) and (\ref{Second_CommutEq_for_Hamiltonian}) cannot be satisfied by non-trivial finite-dimensional operators.    

\begin{theorem}\label{NoFreeFiniteDimQM_Theorem}
	Assume $\hat{p}$ and $\hat{H}$ are non-zero finite-dimensional Hermitian operators. If $\commut{ \hat{H}, \hat{p} } = 0$ then there is no finite-dimensional linear operator $\hat{x}$ such that $i \commut{ \hat{H}, \hat{x} } = \hat{p}$.
\end{theorem}
\begin{proof}
	Let $\ket{E_n}$ denote a joint eigenvector of $\hat{H}$ and $\hat{p}$ such that $\hat{H} \ket{E_n} = E_n \ket{E_n}$ and $\hat{p} \ket{E_n} = p_n \ket{E_n}$. Assuming that $i \commut{ \hat{H}, \hat{x} } = \hat{p}$, one obtains 
	\begin{align}
		i \bra{E_n} \commut{ \hat{H}, \hat{x} } \ket{E_m} = i (E_n - E_m) \bra{E_n} \hat{x} \ket{E_m} = p_m \langle E_n \ket{E_m}, 
	\end{align}
	$\forall n,m$.
	The special case of this identity, when $n=m$, reads $p_m = 0$, $\forall m$. This contradicts the assumption that $\hat{p}$ is a non-zero matrix.
\end{proof}

Theorem \ref{NoFreeFiniteDimQM_Theorem} physically implies the nonexistence of the ``free'' particle in finite-dimensional quantum mechanics. This statement is also a consequence of probability conservation, closely linked to Stone's theorem. If the wave function's normalization is constant and $\hat{x}$ is finite-dimensional, then the particle's motion is confined to the spatial interval bound by the smallest and largest eigenvalues of $\hat{x}$. This confinement must be realized by some force. Since standalone Eq. (\ref{First_CommutEq_for_Hamiltonian}) is solvable with respect to $\hat{H}$ [see Eq. (\ref{Ehrenfest_Compartible_GeneralHamiltonian}) below], then $i \commut{ \hat{H}, \hat{p} }/\hbar$ from Eq. (\ref{Second_CommutEq_for_Hamiltonian}) can be taken as the definition of this confining force. Note that the Ehrenfest theorems do not hold for a particle in a box even in the infinite-dimensional case \cite{Alonso2000}.

\begin{theorem}\label{NoDiffAsCommutator_Theorem}
	There are no finite-dimensional operators $\hat{A}$ and $\hat{B}$ such that 
	\begin{align}\label{NoDiffAsCommutator_Eq1}
		\commut{ f(\hat{A}), \hat{B} } =  f'(\hat{A})
	\end{align}
	for any ``good'' function $f(z)$.
\end{theorem}
\begin{proof}
	Let a function of an operator be defined as  (see, e.g., Ref. \cite{Nazaikinskii1992})
	\begin{align}
		f(\hat{A}) \coloneqq \int g(k) e^{ik\hat{A}} dk, \quad f'(\hat{A}) \coloneqq \int ik g(k) e^{ik\hat{A}} dk, \notag
	\end{align}
	where $g(k)$ is the Fourier transform of a function $f(z)$. The correctness of Eq. (\ref{NoDiffAsCommutator_Eq1}) implies 
	\begin{align}
		\commut{ e^{ik\hat{A}}, \hat{B} } = ik e^{ik\hat{A}}.
	\end{align}
	Differentiating the last identity with respect to $k$ and setting $k=0$, we reach the contradiction $\commut{ \hat{A}, \hat{B} } = 1$.
\end{proof}
Theorem \ref{NoDiffAsCommutator_Theorem} establishes not only the nonsolvability of  Eq. (\ref{Second_CommutEq_for_Hamiltonian}) but also two nonequivalent definitions ($\hat{F} = - U'(\hat{x})$ and $\hat{F} = i\commut{U(\hat{x}), \hat{p}}/\hbar$) of a potential force $\hat{F}$ in a finite-dimensional framework. The latter definition should be preferred due to the comment after Theorem \ref{NoFreeFiniteDimQM_Theorem}.

As thoroughly elaborated in Ref. \cite{Bondar2011c}, the Liouvillian $\hat{L}$, the classical dynamics' generator, is a solution of the equations: $im\commut{ \hat{L}, \hat{x} } = \hat{p}$ and $i\commut{\hat{L}, \hat{p}} = -U'(\hat{x})$ with commuting $\hat{x}$ and $\hat{p}$. The following statement establishes that finite-dimensional classical mechanics is more convoluted than the quantum analog:
\begin{theorem}\label{NoFiniteDimClassicalMecahnics}
	Let $\hat{x}$ and $\hat{p}$ be non-zero finite-dimensional Hermitian operators with $\commut{\hat{x}, \hat{p}} = 0$. There is no finite-dimensional operator $\hat{L}$ such that $i\commut{ \hat{L}, \hat{x} } = \hat{p}$. 
\end{theorem}
\begin{proof}
	Up to the notation ($\hat{H} \leftrightarrow \hat{x}$ and $\hat{x} \leftrightarrow \hat{L}$), this theorem is the same as Theorem \ref{NoFreeFiniteDimQM_Theorem}. 
\end{proof}

\section{Inconsistencies and Numerical Calculations}

Finite-dimensional quantum mechanics' most important application is in numerical simulations of the Schr\"{o}dinger equation. We shall now show that the accuracy of numerical calculations appears not to depend on whether a finite-dimensional Hamiltonian obeys the Ehrenfest theorem (\ref{EhrenfestTh1}). Even though we analyze a specific system with a few different choices for its Hamiltonian, the conclusions are universal because the root of the problem lies in the nonexistence of bounded operators obeying the canonical commutation relation. 

Assuming $\hat{x}$ is Hermitian and has a non-degenerate spectrum, the solution of Eq. (\ref{First_CommutEq_for_Hamiltonian}) is
\begin{align}\label{Ehrenfest_Compartible_GeneralHamiltonian}
	\bra{x_k} \hat{H}^{\star} \ket{x_l} = \left\{
		\begin{array}{ll}
			U(x_k) + C & \mbox{if } k=l \\
			\frac{i \hbar \bra{x_k} \hat{p} \ket{x_l} }{m(x_k - x_l)} &\mbox{otherwise}
		\end{array}
	\right.
\end{align}
where $\hat{x} \ket{x_l} = x_l \ket{x_l}$ and $C$ is an arbitrary real constant. For any definition of the momentum operator $\hat{p}$, Eq. (\ref{Ehrenfest_Compartible_GeneralHamiltonian}) provides the most general Hamiltonian obeying the first Ehrenfest theorem [Eq. (\ref{EhrenfestTh1})].

Choosing $\hat{x}$ with equally spaced eigenvalues (with the step $a$) and utilizing the approximation $d \Psi(x_k) /dx \approx [ \Psi(x_{k+1}) - \Psi(x_{k-1}) ] / (2a)$, we can define the momentum operator: $\bra{x_k} \hat{p}_{fd} \ket{x_l} = -i\hbar\left( \delta_{l,k+1} - \delta_{l,k-1} \right)/(2a)$. Then, the corresponding Hamiltonian from Eq. (\ref{Ehrenfest_Compartible_GeneralHamiltonian}), 
\begin{align}\label{FiniteDiff_approx_Hamiltonian}
	\bra{x_k} \hat{H}^{\star}_{fd} \ket{x_l} = \frac{-\hbar^2}{2ma^2} \left( \delta_{l,k+1} + \delta_{l,k-1} \right) +  [U(x_k) + C]\delta_{l,k},
\end{align}
satisfying the Ehrenfest theorem (\ref{EhrenfestTh1}), takes the form of the simplest finite-difference approximation for the differential operator: $-\hbar^2/(2m) d^2/dx^2 + U(x)$. 

Let $\hat{H}_{mub}^{\star}$ denote the Hamiltonian (\ref{Ehrenfest_Compartible_GeneralHamiltonian}) after substituting $\hat{p}$ from Eq. (\ref{Deff_PX_as_MUBS}), and $\hat{H}_{mub}$ be the Hamiltonian (\ref{Infinite_Dim_Form_Hamiltonian}) with the same $\hat{p}$. Contrary to $\hat{H}_{mub}^{\star}$, the matrix $\hat{H}_{mub}$ does not satisfy the Ehrenfest theorem (\ref{EhrenfestTh1}).

Figures \ref{Fig_SpectraComparison} and \ref{Fig_EigenvectorComparison} depict the eigenvalues and the twentieth eigenvector of the Hamiltonians $\hat{H}^{\star}_{fd}$, $\hat{H}_{mub}^{\star}$, and  $\hat{H}_{mub}$ for the oscillator, 
\begin{align}\label{SingularHarmonicOscil_Hamiltonian}
	U(x) = \frac{\omega^2}{8} x^2  + \frac{g}{4x^2}, \qquad 0 \leqslant x < \infty, 
\end{align}
which is exactly solvable in the infinite-dimensional case \cite{Calogero1969}. The values of the coordinate step, $a$, for $\hat{H}_{mub}^{\star}$ and $\hat{H}_{mub}$ are selected to minimize the largest eigenvalues of each Hamiltonian, which leads to the spectra of  $\hat{H}_{mub}^{\star}$ and $\hat{H}_{mub}$ (for a finite $N$) best fitting the exact energy in the infinite-dimensional case [see Figs. \ref{Fig_SpectraComparison}(b) and \ref{Fig_SpectraComparison}(c)]. However, the same recipe applied to $\hat{H}_{fd}^{\star}$ does not provide a good fit, and the step size for $\hat{H}_{fd}^{\star}$ is manually adjusted to match the exact energies. This exercise  demonstrates [Fig. \ref{Fig_EigenvectorComparison}(a)] that even though the Hamiltonian $\hat{H}_{fd}^{\star}$ reproduces the exact energy better than either $\hat{H}_{mub}^{\star}$ or $\hat{H}_{mub}$, it yields poor quality excited states. Finally, note that the Hamiltonian $\hat{H}_{mub}^{\star}$, satisfying Ehrenfest theorem (\ref{EhrenfestTh1}), provides results of the same quality as $\hat{H}_{mub}$, violating Eq.  (\ref{EhrenfestTh1}). These observations indicate that there is no correlation between the  numerical calculations' accuracy and the Ehrenfest theorem's validity. Note, however, that any disagreement between the exact and numerical results in Figs. \ref{Fig_SpectraComparison} and \ref{Fig_EigenvectorComparison} will exponentially increase with the number of spatial dimensions.

\begin{figure}
	\begin{center}
		\includegraphics[scale=0.5]{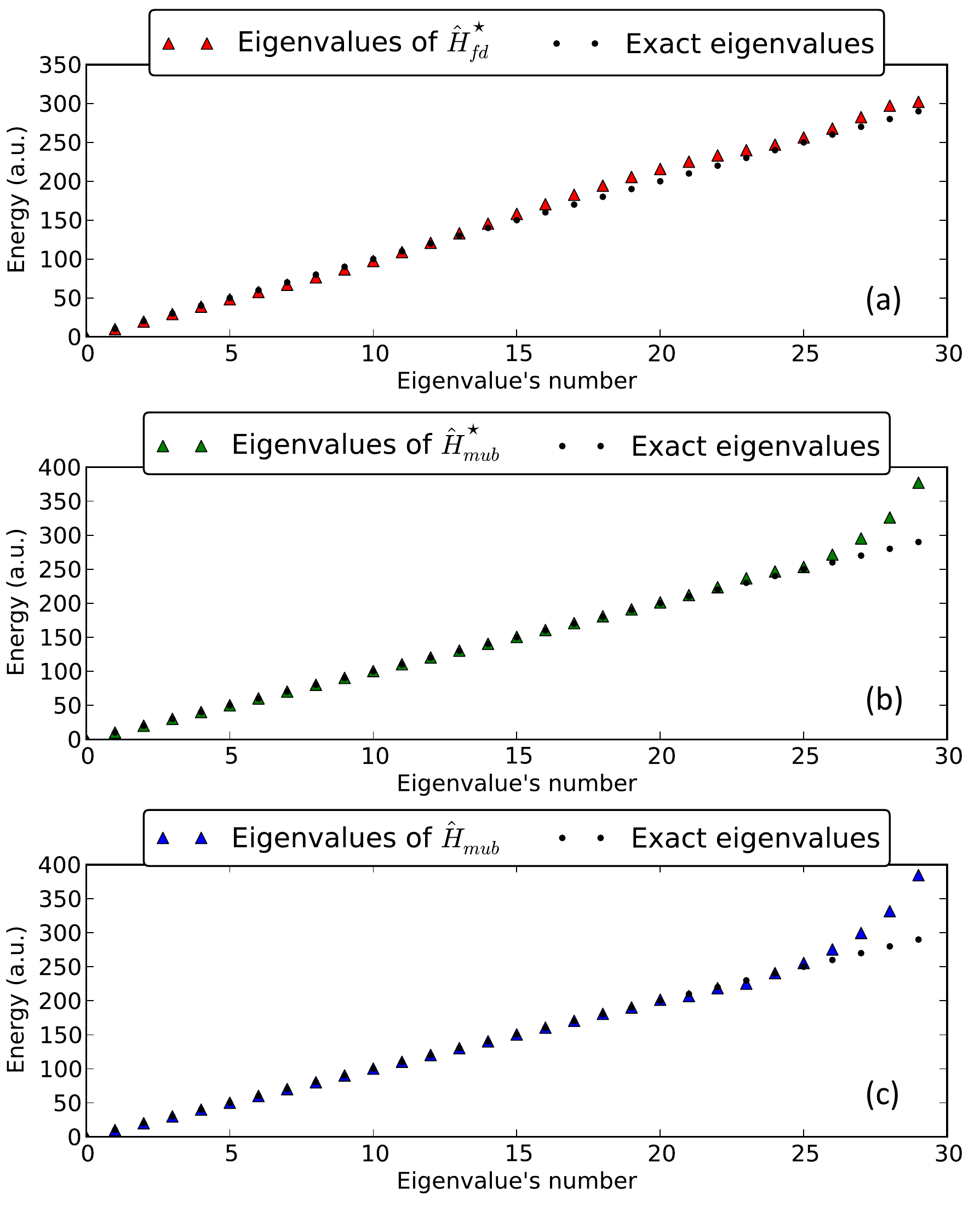}
		\caption{(Color on-line) The exact energies in the infinite-dimensional case \cite{Calogero1969} vs eigenvalues of 30-dimensional Hamiltonians for the singular harmonic oscillator (\ref{SingularHarmonicOscil_Hamiltonian}). $\hat{H}_{fd}^{\star}$ is given by Eq. (\ref{FiniteDiff_approx_Hamiltonian}). $\hat{H}_{mub}^{\star}$ is from Eq. (\ref{Ehrenfest_Compartible_GeneralHamiltonian}) with the momentum operator from Eq. (\ref{Deff_PX_as_MUBS}). $\hat{H}_{mub}$ is from Eq. (\ref{Infinite_Dim_Form_Hamiltonian})  with the momentum operator from Eq. (\ref{Deff_PX_as_MUBS}).  The Hamiltonians $\hat{H}_{fd}^{\star}$ and $\hat{H}_{mub}^{\star}$ satisfy the Ehrenfest theorem (\ref{EhrenfestTh1}) while $\hat{H}_{mub}$ does not.  The values of the spatial step, $a$, are (a) $0.09$, (b) $0.145$, (c) $0.145$. The other parameters are $\hbar=m=1$, $g=1$, $\omega=10$, and $N=30$.}
		\label{Fig_SpectraComparison}
	\end{center}
\end{figure}
\begin{figure}
	\begin{center}
		\includegraphics[scale=0.5]{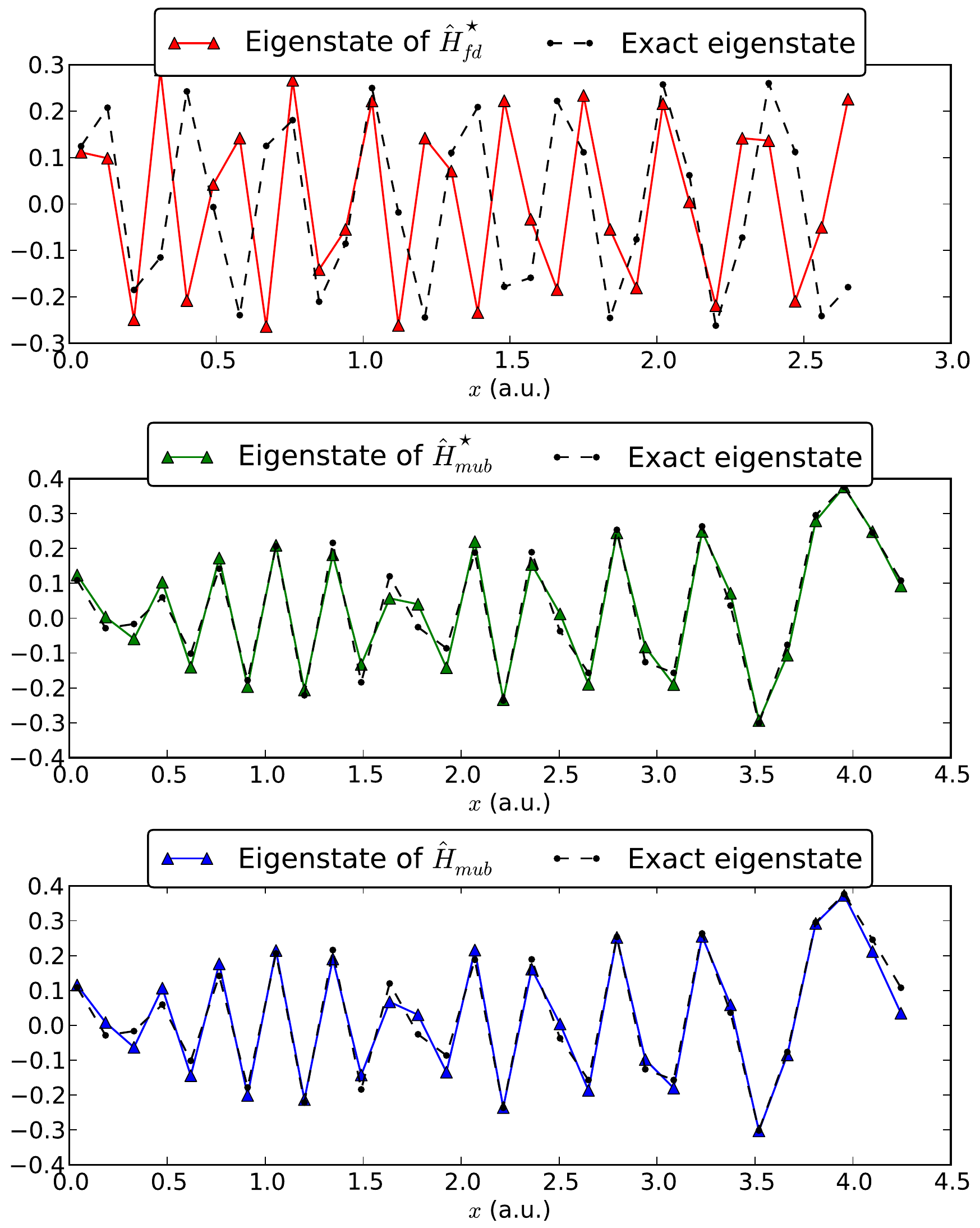}
		\caption{(Color on-line) The exact twentieth eigenstate in the infinite-dimensional case \cite{Calogero1969} vs twentieth eigenvectors of 30-dimensional Hamiltonians for the oscillator (\ref{SingularHarmonicOscil_Hamiltonian}). See Fig. \ref{Fig_SpectraComparison} regarding the parameters' values and notation's description.}
		\label{Fig_EigenvectorComparison}
	\end{center}
\end{figure}

\section{Non-Hermitian Hamiltonian as an Unsuccessful Attempt to Fix Inconsistencies}

Non-Hermitian Hamiltonians are widely employed to model resonant and unbound states (see, e.g., Refs. \cite{Moiseyev2011, Graefe2011} and references therein) as well as a handy trick to avoid numerical artifacts related to the finiteness of spatial grids (see, e.g., Ref. \cite{Manolopoulos2002}). The Hamiltonian's non-Hermitianity leads to non-conservation of the wave function's normalization. The normalization decreasing over time is interpreted as a quantum particle leaving the spatial interval bound by the coordinate's smallest and largest eigenvalues. Therefore, it might be anticipated that this will at least restore the free particle's notion. Despite great utility, non-Hermitian Hamiltonians do not resolve the conceptual difficulties, as we shall demonstrate now.

First, re-applying ODM to Eqs. (\ref{EhrenfestTh1}) and (\ref{EhrenfestTh2}) with the condition $\hat{H} \neq \hat{H}^{\dagger}$, the following generalizations of Eqs. (\ref{First_CommutEq_for_Hamiltonian}) and (\ref{Second_CommutEq_for_Hamiltonian}) are obtained:
\begin{align}
	im ( \hat{H}^{\dagger} \hat{x} - \hat{x} \hat{H} ) = \hbar \hat{p}, \quad 
	i (  \hat{H}^{\dagger} \hat{p} - \hat{p} \hat{H} ) = -\hbar U'(\hat{x}).   
\end{align}
\begin{lemma}\label{NonHermitian_Lemma1}
	Assume that $\hat{x}$, $\hat{p}$, and $\hat{H}$ are bounded. If $\sigma(\hat{H}^{\dagger}) \cap \sigma(\hat{H}) = \emptyset$ and
	\begin{align}
		& im ( \hat{H}^{\dagger} \hat{x} - \hat{x} \hat{H} ) = \hbar \hat{p}, \label{First_OperEq_for_NonHer_FreeHamiltonian} \\
		& \hat{H}^{\dagger} \hat{p} = \hat{p} \hat{H}, \label{Second_OperEq_for_NonHer_FreeHamiltonian}
	\end{align}
	then $\hat{p} = \hat{x} = 0$. Here, $\sigma(\hat{A})$ denotes the spectrum of an operator $\hat{A}$.
\end{lemma}
\begin{proof}
	First, we employ the Rosenblum theorem \cite{Rosenblum1956} (see Theorem 9.3 in Ref. \cite{Bhatia1997}) to Eq. (\ref{First_OperEq_for_NonHer_FreeHamiltonian}) and then employ Eq. (\ref{Second_OperEq_for_NonHer_FreeHamiltonian}) to show 
	\begin{align}
		im\hat{x} =& \frac{\hbar}{2\pi i} \oint_{\Gamma} d\xi ( \hat{H}^{\dagger} - \xi )^{-1} \hat{p} ( \hat{H} - \xi )^{-1} \notag\\
			=&  \frac{\hbar}{2\pi i} \hat{p} \oint_{\Gamma} d\xi ( \hat{H} - \xi )^{-2}  = 0,
	\end{align}
	where $\Gamma$ denotes a union of closed contours in the complex plane with total winding numbers $1$ around $\sigma(\hat{H}^{\dagger})$ and $0$ around $\sigma(\hat{H})$.	
\end{proof}

\begin{theorem}\label{NonHermitian_Theorem1}
	Assume i) $\hat{H}$, $\hat{x}$, and $\hat{p}$ are finite-dimensional operators satisfying Eqs. (\ref{First_OperEq_for_NonHer_FreeHamiltonian}) and (\ref{Second_OperEq_for_NonHer_FreeHamiltonian}); ii) $\lim_{t \to +\infty} \left| \langle \Psi(t) \ket{\Psi(t)} \right| < \infty$, where $\ket{\Psi(t)}$ is a solution of Eq. (\ref{Sch_Eq_General_Form}); iii) the eigenvectors $\{ \ket{E_n} \}_{n=1}^N$ ($\hat{H}\ket{E_n} = E_n \ket{E_n}$, $n \in I \coloneqq \{1,2,\ldots,N\}$) span the entire Hilbert space. Then, $\hat{p} = 0$.
\end{theorem}
\begin{proof}
	Suppose $\hat{p} \neq 0$. According to Lyapunov stability theory \cite{Lyapunov1992}, the second assumption implies that $\Im(E_n) \leqslant 0$, $\forall n\in I$. Lemma \ref{NonHermitian_Lemma1} guarantees that $\hat{H}$ has at least a single real eigenvalue. Introduce $R \coloneqq \{ n \in I \, | \, \Im(E_n) = 0 \} \neq \emptyset$ -- the set of the real eigenvalues' indices. 
	``Sandwiching'' Eqs. (\ref{First_OperEq_for_NonHer_FreeHamiltonian}) and (\ref{Second_OperEq_for_NonHer_FreeHamiltonian}) leads to
	\begin{align}
		& im ( E_k^* - E_l ) \bra{E_k} \hat{x} \ket{E_l} = \hbar \bra{E_k} \hat{p} \ket{E_l}, \label{NonHermTh1_Eq1}\\
		& ( E_k^* - E_l ) \bra{E_k} \hat{p} \ket{E_l} = 0, \qquad \forall k,l \in I. \label{NonHermTh1_Eq2}
	\end{align}
	According to Eq. (\ref{NonHermTh1_Eq2}), $\bra{E_k} \hat{p} \ket{E_l}$ may be non-zero only for $(k,l) \in Q \coloneqq \{ (k,l) \in R \times R \, | \, E_k = E_l \}$. However, Eq. (\ref{NonHermTh1_Eq1}) enforces $\bra{E_k} \hat{p} \ket{E_l} = 0$, $\forall (k,l) \in Q$. Therefore, we reach the contradiction $\hat{p} = 0$.
\end{proof}

Theorem \ref{NonHermitian_Theorem1}, being a generalization of Theorems \ref{NoFreeFiniteDimQM_Theorem} and \ref{NoFiniteDimClassicalMecahnics}, disproves the existence of the free quantum and classical particle within a non-Hermitian setting.

\section{Outlook}

Utilizing ODM \cite{Bondar2011c}, we demonstrated that {\it all} finite-dimensional representations of quantum and classical dynamics  violate the second Ehrenfest theorem [Eq. (\ref{EhrenfestTh2})], while the first Ehrenfest theorem [Eq. (\ref{EhrenfestTh1})] may be satisfied under special circumstances [Eq. (\ref{Ehrenfest_Compartible_GeneralHamiltonian})].  Nonexistence of the free particle case and the ambiguity in defining potential forces were also established. The fundamental reason behind these inconsistencies is the absence of bounded operators obeying the canonical commutation relation. These conclusions fundamentally bound the accuracy achievable by the current paradigm of numerical simulations.

The unveiled inconsistencies may be circumvented in some circumstances. A quantum mechanical system simulated on a continuous variable (i.e., infinite-dimensional) quantum computer \cite{braunstein2010quantum} would not be affected by such difficulties. Additionally, there is a solution not using a quantum computer. In particular, a key element in proving theorems \ref{NoFreeFiniteDimQM_Theorem} and \ref{NoFiniteDimClassicalMecahnics} is the non-existence of a finite real number $r$ such that $0\cdot r \neq 0$. However, this equation has a solution if infinitesimal and infinitely-large numbers are included into the set of real numbers. Thus, some of the no-go theorems may be avoided by utilizing nonstandard analysis \cite{robert2003nonstandard, Almeida2004}, where such extensions are rigorously implemented. An adaptation of the latter approach to physical simulations would be a computational paradigm shift.

We hope that the current work challenges a widespread belief that finite-dimensional quantum mechanics is a straightforward discretization of the corresponding continuous formulae without conceptual consequences, and further motivates exploration of the multifaceted dichotomy between finite- and infinite-dimensional cases. Attesting to this point of view, a recent paper \cite{Fritz2012} contains examples of physically relevant relations realizable only in finite-dimensions and nonexistent in the infinite-dimensional Hilbert space.

\acknowledgments The authors thank Michael Spanner, Serguei Patchkovskii, and Tobias Fritz for valuable comments. Financial support from NSF and ARO is acknowledged.

\bibliography{finite_dim_qm}
\end{document}